\documentclass[a4paper,UKenglish,cleveref, autoref, thm-restate]{lipics-v2021}
\usepackage[utf8]{inputenc}
\usepackage[T1]{fontenc}
\usepackage{amsmath, amsfonts, amssymb, amsthm, mathtools}
\usepackage{graphicx}
\usepackage{algorithm, algpseudocode, calc}
\usepackage{xspace}
\usepackage{tikz}
\usepackage{pgfgantt}
\usetikzlibrary {positioning, shapes, backgrounds,arrows,calc}
\usepackage{bm}
\usepackage{clipboard}
\usepackage[colorinlistoftodos]{todonotes}
\usepackage{mathtools}
\usepackage{multirow}
\usepackage[capitalise]{cleveref}

\usepackage{multicol}
\usepackage{graphicx}

\usepackage{epstopdf}
\usepackage{tabularx}
\usepackage{mathtools}
\usepackage{caption}
\usepackage{subcaption}
\usepackage{enumitem}
\usepackage{verbatim}
\usepackage{xspace}
\usetikzlibrary{matrix}
\usepackage{pdfpages}
\usepackage{babel} 
\usepackage{subcaption}

\usepackage[
singlelinecheck=false,
font=small,
]{caption}

\DeclarePairedDelimiter\ceil{\lceil}{\rceil}

\algnewcommand{\IfThenElse}[3]{
	\State \algorithmicif\ #1\ \algorithmicthen\ #2\ \algorithmicelse\ #3}

\algnewcommand{\IfThen}[2]{
	\State \algorithmicif\ #1\ \algorithmicthen\ #2\ }

\newcommand{\N}{\mathbb{N}}

\newcommand{\FPTAS}{\textnormal{\sffamily FPTAS}\xspace}

\newcommand{\NPClass}{\textnormal{\sffamily NP}\xspace}

\newcommand{\PClass}{\textnormal{\sffamily P}\xspace}

\newcommand{\CostColoring}{\textnormal{\sffamily Cost\xspace Coloring}\xspace}

\newcommand{\ChromaticNumber}{\textnormal{\sffamily Chromatic Number}\xspace}

\newcommand{\completebipartite}{\textit{complete bipartite}}
\newcommand{\emptygraph}{\textit{empty}}
\newcommand{\bipartite}{\textit{bipartite}}

\newcommand{\Osymbol}{\textnormal{O}}
\newcommand{\osymbol}{\textnormal{o}}
\newcommand{\omegasymbol}{\omega}

\newcommand{\sopt}{S_{opt}}
\newcommand{\salg}{S_{alg}}
\newcommand{\sch}{S}
\newcommand{\cmaxcost}{C_{\max}}
\newcommand{\optcmaxcost}{C_{\max}^{*}}

\newcolumntype{C}[1]{>{\centering\arraybackslash}p{#1}}

\newcommand{\OnePreExt}{\textnormal{\sffamily 1-PrExt}\xspace}

\algblockdefx[SpecCase]{SpecCase}{EndSpecCase}[1][]{}{}
\algtext*{EndSpecCase}

\newcommand{\gnnp}{\mathbb{G}_{n,n,p(n)}}
\newcommand{\psum}{\sum p_j}
\newcommand{\optcmaxcostlb}{C_{\max}^{**}}
\newcommand{\job}{J}

\title{Scheduling on uniform and unrelated machines with bipartite incompatibility graphs}
\author{Tytus Pikies}{Dept. of Algorithims and System Modelling, ETI Faculty, \\ Gda\'nsk University of Technology, 11/12 Gabriela Narutowicza Street, 80-233 Gda\'nsk, Poland}{tytpikie@pg.edu.pl}{}{}
\author{Hanna Furma\'nczyk}{Institute of Informatics, Faculty of Mathematics, Physics and Informatics,\\University of Gda\'nsk, 57 Wita  Stwosza Street, 80-309 Gda\'nsk, Poland}{hanna.furmanczyk@ug.edu.pl}{}{}
\authorrunning{T. Pikies and H. Furma\'nczyk}
\nolinenumbers


\Copyright{Tytus P. and Hanna F.}


\EventEditors{}
\EventNoEds{2}
\EventLongTitle{}
\EventShortTitle{}
\EventAcronym{}
\EventYear{}
\EventDate{}
\EventLocation{}
\EventLogo{}
\SeriesVolume{}
\ArticleNo{}

\begin{document}
	\maketitle
\begin{abstract}
    In this paper the problem of scheduling of jobs on parallel machines under incompatibility relation is considered. 	
    In this model a binary relation between jobs is given and no two jobs that are in the relation can be scheduled on the same machine. 
    In particular, we consider job scheduling under incompatibility relation forming bipartite graphs, under makespan optimality criterion, on uniform and unrelated machines.
    
	We show that no algorithm can achieve a good approximation ratio for uniform machines, even for a case of unit time jobs, under $\PClass \neq \NPClass$. 
	We also provide an approximation algorithm that achieves the best possible approximation ratio, even for the case of jobs of arbitrary lengths $p_j$, under the same assumption.
	Precisely, we present an $\Osymbol(n^{1/2-\epsilon})$ inapproximability bound, for any $\epsilon > 0$; and  $\sqrt{\psum}$-approximation algorithm, respectively.
	To enrich the analysis, bipartite graphs generated randomly according to Gilbert's model $\mathcal{G}_{n,n,p(n)}$ are considered. 		
	For a broad class of $p(n)$ functions we show that there exists an algorithm producing a schedule with makespan almost surely at most twice the optimum. 
	Due to our knowledge, this is the first study of randomly generated graphs in the context of scheduling in the considered model.
	
    For unrelated machines, an FPTAS for $R2|G = \bipartite|\cmaxcost$ is provided. 
	We also show that there is no algorithm of approximation ratio $\Osymbol(n^bp_{\max}^{1-\epsilon})$, even for $Rm|G = \bipartite|\cmaxcost$ for $m \ge 3$ and any $\epsilon > 0$, $b > 0$, unless $\PClass = \NPClass$.
\end{abstract}

	\section{Introduction}
	Imagine that we are a government of some country during an epidemic.   
	At our disposal are medical facilities.
	Due to various reasons the medical facilities can process different number of patients daily.
	Imagine that we want to inoculate a population with ethnic tensions, not exacerbating the tensions, using these facilities.  
	The population consists of $2$ distinct ethnic/religious/ideological groups.
	For any of these groups some of its members may be in conflict with some members of the other group.
	Of course, inter group tensions also exists, but they are not of our interest.
	We would like to inoculate the population at facilities, minimizing the time of inoculation campaign, provided that for any $2$ members coming from the distinct groups and assigned to the same facility there is no conflict between them.
	
	In the literature there are many variants of scheduling problems considered.
	In this paper we consider one of classical models of jobs scheduling enriched with an additional constraint.
	Given a set of parallel machines $M=\{M_1,\ldots,M_m\}$, a set of jobs $J=\{\job_1, \ldots, \job_n\}$, and a simple graph $G = (J, E)$ corresponding to the incompatibility relation, our goal is to assign each of the jobs to a machine in order to minimize the makespan (denoted by $\cmaxcost$).
	The constraint is that the jobs assigned to any machine have to form independent set in $G$. 
	In particular, in this paper we consider the case when $G$ is bipartite.
	 
	We assume that with each $\job_j \in J$ there is associated a natural number, called processing requirement, denoted $p_j$.
	By $\psum$ we denote $\sum_{j\in[n]}p_j$.
	By $p_{\max}$ we denote $\max_{j \in [n]}p_j$.
	We consider the following three variants of machine environment.
	\begin{itemize}
		\item 	\emph{Identical} machines, denoted by $\alpha = P$ in the three field notation $\alpha|\beta|\gamma$ of \cite{lawler1982recent}.
		        The processing time of each job $\job_j$ is equal to $p_j$, for each $j \in [n]$. 
		\item   \emph{Uniform} machines, denoted by $Q$.
		        In this variant a machine $M_i$ runs with speed $s_i\geq 1$, $i\in[m]$.
		        The processing time of $\job_j$ on the machine $M_i$ is equal to $p_j/s_i$. 
				For simplicity of presentation it is assumed that the machines are sorted in a non-increasing order of speeds, i.e. $s_1 \geq \cdots \geq s_m$. 
		\item 	For \emph{unrelated} machines, denoted by $R$, the processing time of a job depends on machine in arbitrary way.
				In this variant there are given $mn$ values $p_{i,j}$ for each  $i \in [m]$ and $j \in [n]$, defining the processing time of $\job_j$ on $M_i$. 
	\end{itemize}
    By $C_j(S)$ we denote the completion time of $\job_j$ in a schedule $S$.
    For brevity, we shorten the notation to $C_j$ when $S$ is clear from the context.
	By $\cmaxcost(S)$ we understand $\max_{\job_j \in J}C_j(S)$, i.e. a makespan of $S$.
	By $\optcmaxcost$ we denote $\min_{S \in \mathcal{S}} \cmaxcost(S)$, where $\mathcal{S}$ is the set of all proper schedules for the instance.
	Also, due to widely used association of the job scheduling under $\cmaxcost$ criterion and bin packing \cite{HochbaumSAPolynomial88} we use the term \emph{capacity} of a machine in a given time.
	It is the maximum processing requirement that the machine can handle in the given time. 
	Due to this we also talk about covering a set of jobs $J'$ by a set of machines $M'$ in a given time. 
	It means that the sum of capacities of $M'$ in a given time is at least as $\sum_{\job_j \in J'}p_j$.
	By a \emph{coloring} of a graph $(V, E)$ we understand a function $f : V \rightarrow \mathcal{C}$, such that $f(v_1) \neq f(v_2)$ for each $\{v_1, v_2\} \in E$.
	The set $\mathcal{C} = \{c_1, \ldots, \}$ is finite or not, depending on the problem under consideration.
	By a weighted graph we understand a tuple $(V, E, w)$ where $(V, E)$ forms a graph and $w: V \rightarrow \N$ is the weight of the vertex.
	For other notions, of graph theory and scheduling theory we refer the reader to standard textbooks, \cite{DiestelGraphTheory} and \cite{Brucker1999Scheduling}, respectively.

	Clearly, we can model our motivational problem using the presented framework.
	The people are modeled as jobs.
	A conflict between two persons is modeled as an edge between the corresponding jobs/verticies in the incompatibility graph.
	The medical facilities are modeled as machines and the daily capacity of centers as the speeds of the machines. The division into groups becomes a division into parts of the graph.
	The criteria corresponds directly.
	
	In the next section we discuss the related work.
	In \cref{section:uniform} we study uniform machines. 
	We prove that there is no algorithm with a good worst-case approximation ratio for $Q3|G = \bipartite, p_j=1|\cmaxcost$.
	Also, we provide an algorithm for $Q|G = \bipartite|\cmaxcost$ which in a sense has the best possible approximation ratio.
	Despite the inapproximability result, in \cref{section:random} we enrich our analysis by considering data generated randomly, according to Gilbert's model $\mathcal{G}_{n,n,p(n)}$ \cite{Gilber59Random}.
	For such a data we propose an algorithm that constructs a schedule $S$ for which $\cmaxcost(S) \le 2\optcmaxcost$ a.a.s. (asymptotically almost surely).
	In \cref{section:unrelated} we complete our study by considering unrelated machines model, presenting an FPTAS for $R2|G= \bipartite|\cmaxcost$ and inapproximability results for $R3|G = \completebipartite|\cmaxcost$.
	Finally, in the last section we briefly discuss some open problems.

	Hence, we establish the complexity status of scheduling on uniform machines. 
	This result justifies the research on families of graphs more restrictive than just being bipartite, like bipartite with bounded degree considered in \cite{FurmanczykKSchedulingOfUnit2017}, \cite{FURMANCZYK2018210}, and \cite{Pikies_T._Better_2019}.
	Moreover, due to our best knowledge, we present the first study of random graphs within the considered framework of scheduling.
	Finally we establish the status of the scheduling problem on unrelated machines.
	By this, we generalize earlier results on $P2|G=\bipartite|\cmaxcost$ \cite{bodlaender1994scheduling}.

	\section{Related Work}
	First at all, let us notice that when $G = \emptygraph$, i.e. $G$ is a graph without edges, the problem $\alpha|G = \emptygraph|\cmaxcost$, for $\alpha \in \{P, Q, R\}$ is equivalent to $\alpha||\cmaxcost$.
	In it is known that the problem is NP-hard even for two for identical machines \cite{GareyJComputersAndIntractability1979}. 
	For a finite number of unrelated machines, and hence also for uniform machines, there exist FPTASes \cite{HorowitzExactAndApproximate1976} and \cite{JansesnPImprovedApproximation2001}.
	We mention the second result due to much better time and memory complexities.
	When the number of machines is part of the input, the problem for identical machines is strongly NP-hard  \cite{GareyJComputersAndIntractability1979}. 
	For uniform machines and hence also for identical machines, there is a PTAS \cite{HochbaumSAPolynomial88}. 
	Moreover, \cite{DBLP:journals/siamdm/Jansen10} provided an EPTAS for uniform machines, significantly improving the earlier result of  \cite{HochbaumSAPolynomial88}.  
	Further results, in particular approximation algorithms, can be easily found in the literature (for an example cf. \cite{LenstraApproximation1990}).
	For unrelated machines there is no $(\frac{3}{2} - \epsilon)$-approximate algorithm, but there is $2$-approximate algorithm, by the seminal result of \cite{LenstraApproximation1990}.
	
	The problem of job scheduling with incompatibility constraints, as defined in this paper, was introduced for identical machines in \cite{bodlaender1994scheduling}, and the research is still actively conducted.
	The authors of \cite{bodlaender1994scheduling} gave polynomial time approximation algorithms for scheduling problem with incompatibility graph being $k$-colorable and the number of machines $m \ge k + 1$ with approximation ratio depending on the ratio of $m$ to $k$. 
	The authors proved that there is a polynomial time $2$-approximation algorithm for  $P|G=\bipartite|\cmaxcost$ with $m \geq 3$.
	Moreover, they proved that it achieves the best possible approximation ratio, under $\PClass \neq \NPClass$.
	Also, in passing, they proposed an FPTAS for $P2|G = \bipartite|\cmaxcost$. 
	For trees, a subclass of bipartite graphs, they provided $\Osymbol(n\log n)$-time algorithm with approximation ratio equal to $5/3$.
	Finally, they provided an FPTAS for graph of bounded treewidth and fixed number of machines.
    The authors of \cite{FURMANCZYK2018210} proved NP-hardness of scheduling unit-time jobs forming 3-chromatic cubic graph on three uniform machines.
	Moreover, a polynomial time algorithm was given for bipartite cubic graphs and $3$ uniform machines. 	In \cite{Furmanczyk} it was proved NP-hardness for bipartite graphs and four uniform machines.
	In \cite{Pikies_T._Better_2019} there is  2-approximate algorithm for the problem $Qm| p_j = 1, G = bisubquartic|C_{\max}$ presented.
	Some further results for the problem with bipartite incompatibility graph are provided in \cite{MallekBBScheduling2019}. 
	In particular, the authors proved that the problem $Q|G = \completebipartite, p_j=1|\cmaxcost$ is NP-hard under binary encoding of graph, i.e. an encoding where only the number of jobs in each part is given as input data.
	Recently, in \cite{Pikies2020scheduling} results for complete multipartite incompatibility graphs were presented, bringing some exact or approximate algorithms for uniform machines, and inapproximability results for unrelated ones.
	For example, a polynomial time algorithm for the  $Q|G = \completebipartite, p_j=1|\cmaxcost$ under the customary unary encoding was given.
	There are also other papers where $P|G = graph, p_j=1|\cmaxcost$ with other classes of incompatible graphs was considered: cographs, interval graphs and bipartite graphs \cite{BodlaenderJOnTheComplexity93}.
	Also, a closely related problem of partitioning a graph into independent sets of bounded sizes was considered for complements of bipartite graph  and complements of interval graphs \cite{bodlaender1995restrictions}, and for split graphs \cite{lonc1991complexity}.
	For the last problem polynomial time algorithms were presented for these classes of graphs in the respective papers.
	Finally, recently, the problem for incompatibility graph equal to union of disjoint cliques (also called bags) was considered \cite{DasWieseBags2017}. 
	The authors gave a PTAS for identical machines and they proved there can be no $(\log n)^{1/4-\epsilon}$-approximation algorithm for any $\epsilon >0$, assuming that $\NPClass \not \subseteq ZPTIME(2^{(\log n )^{O(1)}})$ for unrelated machines. 
	
	\section{Some auxiliary observations}
	\label{section:related}
	\begin{definition}
		By an inequitable $2$-coloring of a $($possibly disconnected$)$ graph $G$ we understand a $2$-coloring $(V_1', V_2')$ of $G$ such that $V_1'$ is of maximum cardinality $($$V_1'$ is of maximum total weight, for weighted graphs$)$.
	\end{definition}
	Obviously, the inequitable $2$-coloring of bipartite graph $G=(V,E)$ can be computed in $\Osymbol(|V| + |E|)$ time.
		
	\begin{definition}
		The instance of \OnePreExt is given by: a graph $G$, a number $k \ge 3$, a sequence of vertices $(v_1, \ldots, v_k)$.
		The question is if there is a $k$-coloring $f$ of $G$, such that $f(v_1) = c_1, \ldots, f(v_k) = c_k$.
	\end{definition}
	Of course, for general graphs this problem is NP-complete, due to the fact that it is a generalization of \ChromaticNumber.
	We use the problem with fixed $k$ equal to $3$, for which we have the following.
	\begin{theorem}[\cite{bodlaender1994scheduling}]
		\label{theorem:OnePreExt3Bipartite}
		\OnePreExt is NP-complete for bipartite graphs and $k = 3$.
	\end{theorem}
	

	\section{Uniform Machines}
	\label{section:uniform}
	We have the following theorem as a corollary of the further presented \cref{theorem:FPTAS}.
	This is due to the fact that any FPTAS for $R2|G = \bipartite|\cmaxcost$ applied to a suitably prepared instances of $Q2|G = \bipartite, p_j=1|\cmaxcost$ with $\epsilon = \frac{1}{n+1}$ has to return an exact value.

	\newcounter{countertheoremuniform}
	\setcounter{countertheoremuniform}{\value{theorem}}
	\begin{theorem}
		There exists $\Osymbol(n^3)$ time algorithm for $Q2|G = \bipartite, p_j=1|\cmaxcost$.
		\label{theorem:Q2CmaxOpt}
	\end{theorem}
	However, the FPTAS cannot be applied directly, but has to be applied to a set of instances.
	For the details we refer the reader to the appendix.
	
	The situation for more than $2$ machines becomes completely different.
	We prove that there is no $\Osymbol(n^{\frac{1}{2}-\epsilon})$-approximate algorithm for $Qm|G = \bipartite, p_j = 1|\cmaxcost$, where $m \ge 3$.
	
	We start with a simple observation about necessary components.
	The observations follows from straightforward case analyses of possible colorings of the considered components.
	\begin{lemma}
		Consider any graph $G$, and any set $\mathcal{C}=\{c_1,c_2,\ldots\}$ of at least $2$ colors, and any proper coloring of $G$ using these colors.
		Assume that there is a vertex $v$ in $G$ connected to the component presented in \textbf{\cref{figure:blockers1}}.
		Then at least one of the following cases holds:
		\begin{itemize}
			\item $v$ is colored with color other than $c_1$,
			\item or at least $x$ vertices are colored using colors from $\mathcal{C} \setminus \{c_1\}$.
		\end{itemize}
	\end{lemma}
	
	\begin{lemma}
		Consider any graph $G$, and any set $\mathcal{C}=\{c_1,c_2,\ldots\}$ of at least $3$ colors, and any proper coloring using these colors.
		Assume that there is a vertex $v$ in $G$ connected to the component presented in \textbf{\cref{figure:blockers2}}.
		Then at least one of the following cases holds:
		\begin{itemize}
			\item $v$ is colored with color other than $c_2$,
			\item or at least $x'$ vertices are assigned colors from $\mathcal{C} \setminus \{c_1, c_2\}$,
			\item or at least $x$ vertices are assigned colors from $\mathcal{C} \setminus \{c_1\}$. 
		\end{itemize}
	\end{lemma}
	
	\begin{lemma}
		Consider any graph $G$, and any set $\mathcal{C}=\{c_1,c_2,\ldots\}$ of at least $3$ colors, and any proper coloring using these colors.
		Assume that there is a vertex $v$ in $G$ connected to the component presented in \textbf{\cref{figure:blockers3}}. 
		Then  at least of the following cases holds:
		\begin{itemize}
			\item $v$ is colored with color other than $c_3$, 
			\item or at least $x''$ vertices are assigned colors from $\mathcal{C} \setminus \{c_1, c_2, c_3\}$, 
			\item or at least $x'$ vertices are colored with colors from $\mathcal{C} \setminus \{c_1, c_2\}$, 
			\item or at least $x$ vertices are assigned colors from $\mathcal{C} \setminus \{c_1\}$.
		\end{itemize}
	\end{lemma}

	\begin{figure}[h]
		\centering
		\begin{subfigure}[b]{0.25\linewidth}
		\begin{tikzpicture}
		[place/.style={ellipse,draw=black!100,line width=0.3mm,inner sep=0pt,minimum size=6mm},
			phony/.style={ellipse,draw=black!100,line width=0.3mm, dotted,inner sep=0pt,minimum size=5mm},
		myline/.style={line width=0.3mm}]
		\node at ( 0,0) (v1) [place] {$v_1$};
		\node at ( 0.75,0) (v2) [place] {$\ldots$};
		\node at ( 1.5,0) (v3) [place] {$v_{x}$};
		\node at ( 0.75,-1) (v) [phony] {$v$};
		
		\draw [myline,-] (v1.south) --(v.north);
		\draw [myline,-] (v2.south) --(v.north);
		\draw [myline,-] (v3.south) --(v.north);
		
		\end{tikzpicture}
		\caption
		{
			Component $H_1(x)$.
		}
		\label{figure:blockers1}
		\end{subfigure}
		\hspace{0.025\linewidth}
		\begin{subfigure}[b]{0.25\linewidth}
		\begin{tikzpicture}
		[place/.style={ellipse,draw=black!100,line width=0.3mm,inner sep=0pt,minimum size=6mm},
			phony/.style={ellipse,draw=black!100,line width=0.3mm, dotted,inner sep=0pt,minimum size=5mm},
		myline/.style={line width=0.3mm}]
		\node at ( 0,0) (v11) [place] {$v_1$};
		\node at ( 0.75,0) (v21) [place] {$\ldots$};
		\node at ( 1.5,0) (v31) [place] {$v_{x}$};
		\node at ( 0,-1) (v12) [place] {$v_1'$};
		\node at ( 0.75,-1) (v22) [place] {$\ldots$};
		\node at ( 1.5,-1) (v32) [place] {$v_{x'}'$};
		\node at ( 0.75,-2) (v) [phony] {$v$};
		
		\draw [myline,-] (v11.south) --(v12.north);
		\draw [myline,-] (v11.south) --(v22.north);
		\draw [myline,-] (v11.south) --(v32.north);
		\draw [myline,-] (v21.south) --(v12.north);
		\draw [myline,-] (v21.south) --(v22.north);
		\draw [myline,-] (v21.south) --(v32.north);
		\draw [myline,-] (v31.south) --(v12.north);
		\draw [myline,-] (v31.south) --(v22.north);
		\draw [myline,-] (v31.south) --(v32.north);
		
		\draw [myline,-] (v12.south) --(v.north);
		\draw [myline,-] (v22.south) --(v.north);
		\draw [myline,-] (v32.south) --(v.north);
		\end{tikzpicture}
		\caption
		{
			Component $H_2(x', x)$. 
		}
		\label{figure:blockers2}
		\end{subfigure}
		\hspace{0.025\linewidth}
		\begin{subfigure}[b]{0.4\linewidth}
					\begin{tikzpicture}
			[place/.style={ellipse,draw=black!100,line width=0.3mm,inner sep=0pt,minimum size=6mm},
			phony/.style={ellipse,draw=black!100,line width=0.3mm, dotted,inner sep=0pt,minimum size=5mm},
			myline/.style={line width=0.3mm}]
			\node at ( -1.5,0) (v11) [place] {$v_1$};
			\node at ( -0.75,0) (v21) [place] {$\ldots$};
			\node at ( 0.0,0) (v31) [place] {$v_{x}$};
			\node at ( -1.5,-1) (v12) [place] {$v_1'$};
			\node at ( -0.75,-1) (v22) [place] {$\ldots$};
			\node at ( 0.0,-1) (v32) [place] {$v_{x'}'$};
			\node at ( -0.75,-2) (v13) [place] {$v_1''$};
			\node at ( 0.5,-2) (v23) [place] {$\ldots$};
			\node at ( 1.75,-2) (v33) [place] {$v_{x''}''$};
			
			\node at ( 1,-1) (v11bis) [place] {$v_1^*$};
			\node at ( 1.75,-1) (v21bis) [place] {$\ldots$};
			\node at ( 2.5,-1) (v31bis) [place] {$v_{x}^*$};
			
			\node at ( 0.5,-3) (v) [phony] {$v$};
			
			\draw [myline,-] (v11.south) --(v12.north);
			\draw [myline,-] (v11.south) --(v22.north);
			\draw [myline,-] (v11.south) --(v32.north);
			\draw [myline,-] (v21.south) --(v12.north);
			\draw [myline,-] (v21.south) --(v22.north);
			\draw [myline,-] (v21.south) --(v32.north);
			\draw [myline,-] (v31.south) --(v12.north);
			\draw [myline,-] (v31.south) --(v22.north);
			\draw [myline,-] (v31.south) --(v32.north);
			
			\draw [myline,-] (v11bis.south) --(v13.north);
			\draw [myline,-] (v11bis.south) --(v23.north);
			\draw [myline,-] (v11bis.south) --(v33.north);
			\draw [myline,-] (v21bis.south) --(v13.north);
			\draw [myline,-] (v21bis.south) --(v23.north);
			\draw [myline,-] (v21bis.south) --(v33.north);
			\draw [myline,-] (v31bis.south) --(v13.north);
			\draw [myline,-] (v31bis.south) --(v23.north);
			\draw [myline,-] (v31bis.south) --(v33.north);
			
			\draw [myline,-] (v12.south) --(v13.north);
			\draw [myline,-] (v12.south) --(v23.north);
			\draw [myline,-] (v12.south) --(v33.north);
			\draw [myline,-] (v22.south) --(v13.north);
			\draw [myline,-] (v22.south) --(v23.north);
			\draw [myline,-] (v22.south) --(v33.north);
			\draw [myline,-] (v32.south) --(v13.north);
			\draw [myline,-] (v32.south) --(v23.north);
			\draw [myline,-] (v32.south) --(v33.north);
			
			\draw [myline,-] (v13.south) --(v.north);
			\draw [myline,-] (v23.south) --(v.north);
			\draw [myline,-] (v33.south) --(v.north);
			\end{tikzpicture}
			\caption
			{
				Component $H_3(x'', x', x)$.
			}
			\label{figure:blockers3}
		\end{subfigure}
		\caption
		{
			Components used in \cref{theorem:no_approx}.
		}
		\label{figure:blockerstogether}
	\end{figure}

	\newcounter{countertheoremnoapprox}
	\setcounter{countertheoremnoapprox}{\value{theorem}}
	\begin{theorem}
		For any $c > 0$ and $\epsilon>0$, there is no $c(n^{\frac{1}{2}-\epsilon})$-approximate algorithm for $($$Qm|G = \bipartite, p_j = 1|\cmaxcost$, $m \ge 3$$)$, unless $\PClass = \NPClass$.
		\label{theorem:no_approx}
	\end{theorem}
	\begin{proof}
	The proof is similar to a proof of inapproximability of \CostColoring of \cite{Jansen00ApproximationResults}.
	Let us assume that we have an instance $((V, E), (v_1, v_2, v_3))$ of \OnePreExt on bipartite graph.
	The question is weather this $3$-coloring can be extended to the whole graph.
	We identify the color $c_1$ with $M_1$, color $c_2$ with $M_2$, color $c_3$ with $M_3$.
	Similarly, we interpret $V$ as a set of jobs.
	
	Let us prepare an instance of $Q|G = \bipartite, p_j=1|\cmaxcost$ for a given instance of \OnePreExt.
	Let $k \ge 1$ be an integer value dependent on the instance of \OnePreExt.
	We extend $(V, E)$ by connecting $v_1$ with $H_2(kn, 6k^2n)$ and $H_3(1, kn, 6k^2n)$, $v_2$ with $H_1(6k^2n)$ and $H_3(1, kn, 6k^2n)$ and $v_3$ with $H_1(6k^2n)$ and $H_2(kn, 6k^2n)$, we use $6$ components in total.
	Notice that we have $n' = n + 48k^2n + 4kn + 2 \le 54k^2n$ vertices in the new graph.
	We set $s_1 = 49k^2$, $s_2 = 5k$, $s_3 = 1$, $s_4 = s_5 = \ldots = s_m = \frac{1}{kn}$ for $Q|G = \bipartite, p_j=1|\cmaxcost$.
	
	If the answer to the instance of \OnePreExt is 'YES', then there exists an extension of a $3$-coloring.
	Using such an extension we can prepare schedule such that $M_1$ is given at most $48k^2n + n \le 49k^2n$ vertices, $M_2$ is given at most $4kn + n \le 5kn$ vertices, and $M_3$ is given at most $n$ vertices. 
	This means that there exists a schedule of length at most $\max\{\frac{49k^2n}{49k^2}, \frac{5kn}{5k}, \frac{n}{1}\} = n$.
	
	If the answer to the instance of \OnePreExt is 'NO', then in any schedule $\sch$ at least one of the following cases holds:
	\begin{itemize}
		\item $M_4$ or later machine is used, which means that $\cmaxcost(\sch) \ge kn$.
		\item Some component $H_1$ forces to assign $6k^2n$ jobs to $M \setminus \{M_1\}$, by which $\cmaxcost(\sch) \ge kn$.
		\item Some component $H_2$ forces to assign $kn$ jobs to $M \setminus \{M_1, M_2\}$, by which $\cmaxcost(\sch) \ge kn$, or the component forces the previous case.
		\item Some component $H_3$ forces one of the previous cases.
\end{itemize}
	Hence in any case $\cmaxcost(\sch) \ge kn$.
	
	Assume that we have $(cn^{\frac{1}{2} - \epsilon})$-approximate algorithm.
	Applying it to a prepared instance of the scheduling problem for an 'YES' instance of \OnePreExt has to yield a schedule of $\cmaxcost \le n c(54k^2n)^{\frac{1}{2}-\epsilon} \le 8ck^{1-2\epsilon}n^{3/2}$.
	Notice that it is less than $kn$, if $k > (8c\sqrt{n})^{\frac{1}{2\epsilon}}$.
	Hence we can choose $k = \ceil{(8c\sqrt{n})^{\frac{1}{2\epsilon}}} + 1$, which is clearly polynomially bounded in size of instance of \OnePreExt, for fixed $c$ and $\epsilon$.
	Hence, using the hypothetical algorithm with instances for such a $k$ we would be able to distinguish between 'YES' and 'NO' instances.
	\end{proof}

	In the following we prove that there is $\sqrt{\psum}$-approximate algorithm for $Q|G = \bipartite|\cmaxcost$.
	\begin{algorithm}[H]
	\caption{$\sqrt{\psum}$-approximate algorithm for $Q|G = \bipartite|\cmaxcost$}
	\label{algorithm:SqrtPtotalForQBipartiteCmax}
	\begin{algorithmic}[1]
		\Require A bipartite graph $G=(J, E)$, $m$ uniform machines $M =\{M_1, \ldots, M_m\}$.
		\Ensure A schedule with $\cmaxcost \le \sqrt{\psum}\optcmaxcost$.
		\State If $\psum \le 4$, solve the instance by brute force.
		\State Let $I$ be an independent set of the highest weight containing all jobs of processing requirement at least $\sqrt{\psum}$ in $G$, if such a set exist. 
		\State Let $S_1$ be a schedule produced by \cref{alg:FPTASUnrelated} for $M_1, M_2$ and $\epsilon = 1$.
		\If {$I$ exists}
		\State Let $\optcmaxcostlb$ be the smallest time such that:
		
		- the rounded down capacities of $M_1, \ldots, M_{m}$ are at least $\psum$,
		
		- the rounded down capacities of $M_{2}, \ldots, M_{m}$ are at least $\sum_{\job_j \in J \setminus I} p_j$,
		
		- $M_1$ can process $p_{\max}$.
		\State Round down to the nearest integer the capacities of the machines in time $\optcmaxcostlb$.
		\State Take such $k \ge 3$ that the capacities of $M_2, \ldots, M_k$ are at least $\sum_{\job_j \in J \setminus I} p_j$.
		\State Let $(J_1', J_2')$ be an inequitable coloring of $J \setminus I$.
		\State Take the biggest $k'$ such that the sum of capacities of $M_2, \ldots, M_{k'}$ 
		\Statex \hspace{\algorithmicindent}is at most $\sum_{\job_j \in J_1'}p_j$. If this is not possible take $k' = 2$.
		\State Let $S_2$ be as follows:
		
		schedule $J_1'$ on $M_2, \ldots, M_{k'}$ and $J_2'$ on $M_{k'+1}, \ldots, M_{k}$;
		
		schedule $I$ on $M_1, M_{k+1}, \ldots, M_{m}$.
		\EndIf
		\State \Return The best of $S_1$, and $S_2$, if the latter exists. 
	\end{algorithmic}
	\end{algorithm}

	\begin{theorem}
		\cref{algorithm:SqrtPtotalForQBipartiteCmax} is $\sqrt{\psum}$-approximate for $Q|G = \bipartite|\cmaxcost$.
		\label{theorem:Qapproximation}
	\end{theorem}
	\begin{proof}
		Mind that the processing requirements of jobs are positive integers.
		Observe that $\optcmaxcostlb$ is a lower bound on $\optcmaxcost$ and that it can be efficiently calculated. 
		Notice, that it might be the case that $\optcmaxcost > \optcmaxcostlb$.
		Let $\sopt$ be any optimal schedule.
		\begin{itemize}
			\item Assume that $M_3, \ldots, M_m$ are not doing any job in $\sopt$. 
			Clearly, $S_1$ is at most $2$-approximate.
			Hence assume further that the machine $M_3$ is doing at least one job.
			
			\item Assume that $1/s_3 \ge \sqrt{\psum}/s_2$.
			In this case $S_1$ is clearly $\sqrt{\psum}$-approximate, by the previous assumption.
			Hence assume further that $1/s_3 < \sqrt{\psum}/s_2$.
			
			\item Assume that there are no jobs of processing requirement greater than or equal to $\sqrt{\psum}$.
			Hence, set $I$ exist.
			Moreover, due to the guessing, the capacities of $M_2, \ldots, M_{m}$ can cover $J \setminus I$.
			\begin{itemize}
				\item Assume that $k' = 2$, i.e. the capacity of $M_2$ covers $J_1'$ completely. 
				In this case $M_3$ covers $J_2'$ in time $\sqrt{\psum}\optcmaxcostlb$, due to the fact that $\sum_{\job_j \in J_1'} p_j \ge \sum_{\job_j \in J_2'} p_j$ ( by the definition of the inequitable coloring) and due to the bound on $s_2$ to $s_3$ ratio.
				\item Otherwise capacities of $M_2, \ldots, M_{k'}$ cover $J_1'$ in at least $\frac{1}{2}$ and the capacities of $M_{k'+1}, \ldots, M_{k}$ cover $J_2'$ completely.
				Consider two cases:
				\begin{itemize}
					\item At least one capacity of $M_2, \ldots, M_{k'}$ is of size $1$.
					In this case the capacity of $M_2, \ldots, M_{k'}$ is at least $|J_1'|$ and capacity of $M_{k'+1}, \ldots, M_k$ is at least $|J_2'|$, by the fact that the last machines added were of capacity $1$ we were able to cover the jobs exactly.
					Notice that if we extend the completion time by $\sqrt{\psum}$ then in each time slot of size $\sqrt{\psum}$ we can schedule at least one job - hence we can schedule all of them.
					\item All the capacities of $M_2, \ldots, M_{k'}$ are greater than or equal to $2$.
					If the capacity of $M_i$ was $cap_i$ consider a hypothetical capacity $cap_i' = 2cap_i + \sqrt{\psum}$. 
					To achieve such a capacity $M_i$ has to work $\frac{cap_i'}{cap_i} \le \sqrt{\psum}$ times longer, by the assumption that $\psum \ge 4$.
					The new processing times of the machines were increased to at most $\sqrt{\psum}\optcmaxcostlb$ time.
					We say that the capacities were doubled, and there was attached a margin to each capacity.
					Notice that by the fact that $\sum_{i=2}^{k'}cap_i$ were at least $\frac{1}{2}\sum_{\job_j \in J_1'} p_j$, the doubled capacities together with margins allow to schedule all jobs of $J_2'$ by simple list scheduling.
					For $J_2'$ and $M_{k'+1}, \ldots, M_{k}$ the reasoning from previous point holds.
				\end{itemize}
			\end{itemize}
			Now we have to schedule $I$ on $M_1, M_{k+1}, \ldots, M_m$, again by a simple list scheduling.
			Notice that the capacity potentially wasted (perhaps the machines $M_2, \ldots, M_k$ are underutilized) is at most three times the capacity of $M_1$.
			Hence in time $4\optcmaxcostlb$ $M_1$ can schedule such many jobs that the capacity on $M_{k+1}, \ldots, M_m$ is at least the necessary space for the  remaining jobs.
			Clearly, we can schedule them on $M_{k+1}, \ldots, M_{m}$ in at most $\sqrt{\psum}\optcmaxcostlb$ time.
			\item Assume that there are some jobs of processing requirement greater than or equal to  $\sqrt{\psum}$, and there is no single independent set containing them all.
			In this case $S_1$ is $\sqrt{\psum}$-approximate.
			\item Assume that there are some jobs of processing requirement greater than or equal to  $\sqrt{\psum}$ and there is a single independent set containing all of them.
			\begin{itemize}
				\item Assume that $M_{i \ge 2}$ in $\sopt$ does at least one job of processing requirement $\sqrt{\psum}$ - then $S_1$ is $\sqrt{\psum}$-approximate.
				\item Otherwise, notice that all the jobs have to be schedule on $M_1$.
				Hence, by the previous argument the capacity of $M_2, \ldots, M_m$ is sufficient to schedule $J \setminus I$ in time $\sqrt{\psum}\optcmaxcostlb$.
				And similarly, the space wasted on $M_2, \ldots, M_k$ is at most $3$ times the capacity of $M_1$ - hence in time $4\optcmaxcostlb$ $M_1$ can do more than its proper share.
				Hence the machines $M_{k+1}, \ldots, M_m$ can handle the other jobs from $I$ in time less than or equal to $\sqrt{\psum}\optcmaxcostlb$, by the previously presented reasoning.
			\end{itemize}
		\end{itemize}
	\end{proof}
	In passing notice that the algorithm is in a sense optimal.
	An $\Osymbol((\psum)^{1/2 - \epsilon})$-approximate algorithm has to be $\Osymbol(n^{1/2-\epsilon})$-approximate for $Q|G = \bipartite, p_j=1|\cmaxcost$, which is impossible, even for fixed number of $m \ge 3$ machines, unless $\PClass = \NPClass$.

	\newcounter{counterqapxcomplex}
	\setcounter{counterqapxcomplex}{\value{theorem}}
	
	\begin{lemma}
		\cref{algorithm:SqrtPtotalForQBipartiteCmax} has time complexity $\Osymbol(|J|^2 + |J||E| + |M|\log |M|)$.
		\label{lemma:qapxcomplexity}
	\end{lemma}
	\noindent For the discussion about the complexity see the appendix.

	\subsection{Random Bipartite Graphs}
	\label{section:random}
	By the presented observations, the algorithms with good worst-case approximation ratio are very unlikely to exist.
	However, for a particular type of input data the problem might be easier.
	We study the behavior of the further proposed algorithm on instances where input graph $(J, E)$ is constructed according to a random bipartite graph model.
	In accordance with \cite{JansonLR2000RandomGraphs} we define the random graph $\gnnp$ as probability space $(\Omega, \mathcal{F}, \mathbb{P})$ where:
	$\Omega$ is a set of all spanning graphs of $K_{n,n}$, where by graph $G$ spanning graph $K_{n,n}$ we mean a graph $G$ such that $V(G) = V(K_{n,n})$ and $E(G) \subseteq E(K_{n,n})$,
	$\mathcal{F}$ is the set of all subsets of $\Omega$,
	$\mathbb{P}(G) = p(n)^{|E(G)|}(1-p(n))^{\binom{|V(G)|}{2} - |E(G)|}$.
	In this paper we consider functions $p(n)$ that are monotonic and are such that $0 \le p(n) \le 1$ for each $n$.
	We consider three types of functions $p(n)$: $p(n) = \osymbol(\frac{1}{n})$; $p(n) = \frac{a}{n}$, for some value $a$; and $p(n) = \omegasymbol(\frac{1}{n})$.
	We prove that the schedule constructed by the algorithm that we propose a.a.s has constant approximation ratio when applied to a realization of $\gnnp$.

	Many authors considered properties of graphs generated according to Gilbert's model, in particular the size of maximum matching $\mu(G)$ in such a graph.
	This parameter is of critical importance to bound the approximation ratio of the further proposed algorithm.
	Sadly, to the authors best knowledge, even the exact values of expectation of $\mu(G)$ are unknown for functions of form $p(n) = \frac{c}{n}$ for a constant $c$.
	However, before we apply the observations about $\mu(G)$ let us provide some simple observations about properties of inequitable coloring.
	\begin{corollary}
		Let $p(n) = \osymbol(\frac{1}{n})$, and let $(V_1', V_2')$ be an inequitable coloring of $\gnnp$, then a.a.s. $\frac{|V_2'|}{n} = \osymbol(1)$. 
		\label{corollary:concentrationlowp}
	\end{corollary}
	\begin{proof}
		We construct the following upper bound on the size of the smaller part. We take all non-isolated vertices in $V_2$.
		Let us construct the bounds on the number of non-isolated vertices in $V_2$.
		Such a number is certainly at mosts the number of edges in this graph.
		Let the number of edges be $X$, it has binomial distribution with parameters
		\[
		E[X] = n^2p(n),
		\]
		\[
		Var[X] = n^2p(n)(1-p(n)).
		\]

		By Chebyshev's bound we have
		\[
		P((|X - E[X]|) \ge \sqrt{n}) \le np(n)(1 - p(n)) = \osymbol(1).
		\] 
		By the fact that $E[X] = \osymbol(n)$ this means also that a.a.s. the number of edges is $\osymbol(n)$.
		It means that a.a.s. the number of vertices in $V_2'$ is $\osymbol(n)$.
	\end{proof}

	\begin{lemma}
		Let $p(n) = \frac{a}{n}$, for some $a > 0$, and let $(V_1', V_2')$ be an inequitable coloring of $\gnnp$, then a.a.s. $|V_2'| \le  n(1 - (1 + \frac{a}{n})^n) + \osymbol(n)$.
	\end{lemma}
	\begin{proof}
		Let the vertices be denoted as $\{v_1, \ldots, v_{2n}\}$, where first $n$ vertices belongs to first part, and the rest to second.
		To construct a reasonable upper bound on the size of the smaller part we provide a simple heuristic.
		Precisely, for a realization $G = (V_1 \cup V_2; E)$ we take all vertices of $V_2$, except the isolated vertices.
		This simple approach allow us to exploit the fact that events of being independent in $V_2$ are independent.
		Hence we have variables $X_{n+1}, \ldots, X_{2n}$. 
		The variable $X_i$ takes value $1$ if $v_i \in V_2$ is isolated, $0$ otherwise.
		We also construct a variable $X = X_{n+1} + \ldots X_{2n}$.
		Hence $X_i$ has Bernoulli distribution and $X$ has binomial distribution.
		Hence 
		\[
		E[X] = n(1-\frac{a}{n})^n
		\]
		\[
		Var[X] = n(1-\frac{a}{n})^n(1-(1-\frac{a}{n})^n).
		\]
		By Chebyshev's inequality we ge at once that
		\[
		Pr(|X - E[X]| \ge \sqrt{E[X]\log E[X]}) \le \frac{Var[X]}{E[X] \log E[X]} = \frac{1 - (1 - \frac{a}{n})^n}{\log n(1 - \frac{a}{n})^n}.
		\]
		
		The limit of the rightmost function at infinity is $0$. 
		Hence it means that a.a.s. the number of isolated vertices in $V_2$ is at least 
		\[
		n(1 - \frac{a}{n})^n - \osymbol(n).
		\]
		This means that $|V_2'|$ is a.a.s. at most
		\[
		n(1 - (1 - \frac{a}{n})^n) + \osymbol(n).
		\]
	\end{proof}

	The most interesting probability functions are of the form $p(n) = \frac{a}{n}$, for some constant $a$.
	For these we have the following.
	\begin{lemma}[\cite{MastinJGreedyOnline2013}]
		Let $p(n) = \frac{a}{n}$, for some $a > 0$, then a.a.s. $\mu(\gnnp) \ge (1-e^{(e^{-a}-1)})n$.
	\end{lemma}
	We considered the weaker bound for $\mu(G)$ for  $p(n) = \frac{a}{n}$ from that paper, due to the fact that it is much easier to analyze analytically and an application of it does not change the approximation ratio. 

	\newcounter{countervtwotonminusalpha}
	\setcounter{countervtwotonminusalpha}{\value{lemma}}
	Hence, let us bound a.a.s. the size of the smaller part of graph with respect to $n - \alpha(\gnnp)$.
	By choosing a suitable bound on $V_2'$ and high school mathematical analysis we obtain as follows.
	\begin{lemma}
		Let $p(n) = \frac{a}{n}$, for some $a > 0$.
		Then for an inequitable coloring $(V_1', V_2')$ of $\gnnp$, $\frac{|V_2'|}{n - \alpha(\gnnp)} \le 1.6$  a.a.s.		
		\label{lemma:vtwotonminusalpha}
	\end{lemma}
	However, the proof was moved to the appendix due to the space limitations.
	
	Now, we have to prove that for $p(n) = \omegasymbol(\frac{1}{n})$ the matching has size $n(1 - \osymbol(1))$ a.a.s.
	\begin{theorem}[\cite{Bollobas2011RandomGraphs}]
		If $np(n) - \log n \rightarrow \infty$ then a.a.s  $\mu(\gnnp) = n$. 
		\label{theorem:matchingI}
	\end{theorem}

	And as a corollary,
	\begin{corollary}
	Let $p(n) = \Omega(1)$ then a.a.s. $\gnnp$ has matching of size $n$.
	\end{corollary} 	

	Finally we have.
	\begin{theorem}[\cite{Zito2003SmallMaximalMatchings}]
		Let $p(n)n \rightarrow \infty$, then a.a.s. $\beta(\gnnp) > n -\frac{2\log np(n)}{\log (\frac{1}{1-p(n)})}$, where $\beta(G)$ is the size of the smallest maximal matching in $G$.
	\end{theorem}
	
	And as a corollary
	\begin{corollary}
		Let $p(n) = \omega(\frac{1}{n})$ and $p(n) = \osymbol(1)$ then a.a.s. $\gnnp$ has matching of size $(1 - \osymbol(1))n$.
		\label{corollary:matchingII}
	\end{corollary} 
	\begin{proof}
		Observe only that the limit of $\log \frac{1}{1 - p(n)}$ tends to infinity, but $\log np(n) = \Osymbol (\log n)$.
	\end{proof}
	
	Hence, we know how $\mu(\gnnp)$ behaves in each of the ranges.
	All these observations allow us to construct an algorithm for random bipartite graphs.
	\begin{algorithm}[]
	\caption{An algorithm for $Q|G = \gnnp, p_j=1|C_{\max}$}
	\label{algorithm:2CmaxRandomBipartite}
	\begin{algorithmic}[1]
		\Require A bipartite graph $G = (J, E)$, machines $M =\{M_1, \ldots, M_m\}$
		\Ensure A schedule with $\cmaxcost$ a.a.s. at most $2\optcmaxcost$.
		\State Let $(V_1', V_2')$ be any inequitable $2$-coloring of $G$. 
		\State Let $\optcmaxcostlb$ be the least time such that rounded down capacities of the machines are at least $n$.
		\State Take the least $k \le m$ such that: 
		\Statex \hspace{\algorithmicindent} the capacities of $M_2, \ldots, M_k$ are at least $\frac{1}{2}|V_2'|$,
		\Statex \hspace{\algorithmicindent} $k=m$, otherwise.
		\State \Return $M_1, M_{k+1},  \ldots, M_m \leftarrow V_1'; M_{2},  \ldots, M_k \leftarrow V_2'$.
	\end{algorithmic}
	\end{algorithm}	

	\begin{theorem}
		For $Q|G = \gnnp, p_j=1|\cmaxcost$, where $p(n)$ is a monotonic existence-of-edge probability function, \cref{algorithm:2CmaxRandomBipartite} returns a schedule $S$, such that a.a.s. $\cmaxcost(S) \le 2\optcmaxcost$.
	\end{theorem}
	\begin{proof}
		The proof consists of an analysis how for a given sequence of machines the approximation ratio behaves, when $n$ tends to infinity.
		Let $(V_1', V_2')$ be an inequitable coloring of bipartite graph $(J, E)$, the realization of $\mathcal{G}_{n,n,p(n)}$.
		Let also $\salg$ be the schedule produced by the algorithm and $\sopt$ an optimal schedule.

		First consider the case when $p(n) = \osymbol(\frac{1}{n})$.
		In this case by \cref{corollary:concentrationlowp} almost all of vertices are in $V_1'$, hence a.a.s. the machine $M_2$ will be underutilized for $n$ high enough.
		On the other hand the maximum completion time of $V_1'$ on $M \setminus \{M_2\}$ is at most $2$ times longer than completion time of $V_1'$ on $M$.
		
		Consider the case when $p(n) = \frac{a}{n}$.
		Assume that we guessed $\optcmaxcost$.
		Consider two cases:
		\begin{itemize}
			\item The capacities of $M_2, \ldots, M_k$ are at least $\frac{1}{2}|V_2'|$.
			In such a case we can schedule the jobs on $M_2, \ldots, M_k$ in time at most $2\optcmaxcostlb$.
			Moreover, the capacities that were not used on $M_2, \ldots, M_k$ are at most the capacity $M_1$.
			Hence by scheduling $V_1'$ on $M_1, M_{k+1}, \ldots, M_m$, in time at most $2\optcmaxcostlb$ all jobs can be scheduled.
			\item The capacities of $M_2, \ldots, M_m$ are less than $\frac{1}{2}|V_2'|$. 
			However, by \cref{lemma:vtwotonminusalpha}, a.a.s. the $M_2, \ldots, M_m$ are doing at most $1.6$ the minimum number of jobs that has to be assigned to them in any schedule.
			Again, a.a.s. the constructed schedule has length at most twice the optimum.
		\end{itemize}
	
		Consider the cases when $p(n) = \omega(\frac{1}{n})$ and $p(n) = \osymbol(1)$ or $p(n) = \Omega(1)$.
		In these cases the previous analysis holds, because a.a.s. $\mu(\gnnp) \ge n(1 - \osymbol(1))$ , by \cref{theorem:matchingI} and \cref{corollary:matchingII}.
		Hence $(|J| - \alpha(\gnnp)) > (1 - \osymbol(1))n$.
		Hence also $|V_2'| \le |V_2| < (1+\osymbol(1))(n - \alpha(\gnnp))$.
	\end{proof}

\section{Unrelated Machines}
	\label{section:unrelated}	
	We apply the FPTAS considered by \cite{JansesnPImprovedApproximation2001} to obtain the desired FPTAS.

	\begin{theorem}[\cite{JansesnPImprovedApproximation2001}]
		There exists an FPTAS $\Osymbol(n(m/\epsilon)^{\Osymbol(m)})$ time  for $Rm||\cmaxcost$.\label{fptas_unrel}
	\end{theorem}
	We can use an $\FPTAS$ for $R2||\cmaxcost$ to make an $\FPTAS$ for $R2|G = \bipartite|\cmaxcost$.
	Precisely, we propose a simple algorithm that constructs $2$-approximate solution of time $T$.
	Using it we set "private" loads on the machines to emulate the minimum processing times of jobs on each machine, common to all schedules.
	We do this by each connected component separately, therefore reducing the problem $R2|G = \bipartite|\cmaxcost$ into $R2||\cmaxcost$.
	We do so by constructing artificial jobs with processing times equal to the common to all schedules.
	They can be scheduled only on their corresponding machines, on the other machine they will have processing time equal to an unreasonable value, for example, $3T$.
	
	Let us start with this simple $2$-approximate algorithm.
	It consists of two procedures, reduction, and the proper algorithm.

	\begin{algorithm}[H]
		\begin{algorithmic}[1]
			\Require A set of $c$ connected components of bipartite graph $G$: $G_1,\ldots,G_c$, $G_k=(V_1^k \cup V_2^k, E_k)$, $k\in[c]$.
			 A set of $2$ unrelated machines.	
			\Ensure A set of $c$ jobs. Two sequences of values $P', P''$. 
			\For {$k \in [c]$}
			\State Let $p_{i,l}^* := \sum_{v_j \in V_l^k} p_{i, j}$, for $i, l \in [2]$.
			\State Construct the job $\job_{n+k}$ and the pair of values in the following way: 
			\If {$p_{1,1}^* \le p_{1,2}^*$ and $p_{2,2}^* \le p_{2,1}^*$}
			\State Let $p_{1, n+k} := p_{2, n+k} = 0, P'_k:= p_{1,1}^*, P''_k := p_{2,2}^*$.
			\ElsIf {$p_{1,2}^* \le p_{1,1}^*$ and $p_{2,1}^* \le p_{2,2}^*$}
			\State Let $p_{1, n+k} := p_{2, n+k} = 0, P'_k := p_{1,2}^*, P''_k := p_{2,1}^*$.
			\Else
			\State Let $p_{1, n+k} := \max\{p_{1,1}^*, p_{1,2}^*\} - \min\{p_{1,1}^*,p_{1,2}^*\}$.
			\State Let $p_{2, n+k} := \max\{p_{2,1}^*, p_{2,2}^*\} - \min\{p_{2,1}^*,p_{2,2}^*\}$.
			\State Let $P'_k := \min\{p_{1,1}^*,p_{1,2}^*\}, P''_k := \min\{p_{2,1}^*,p_{2,2}^*\}$.
			\EndIf
			\EndFor
			\State \Return $(\{J_{n+1}, \ldots, J_{n+k}\}, P', P'')$.
		\end{algorithmic}
		\caption{Reduction of jobs for $R2|G = \bipartite|\cmaxcost$.}
		\label{alg:unrelatedreduced}
	\end{algorithm}

	\begin{algorithm}[H]
	\begin{algorithmic}[1]
		\Require $G = \bipartite$, $M = \{M_1, M_2\}$	
		\Ensure A $2$-approximate schedule.
		\State $(J' , P', P'') = \cref{alg:unrelatedreduced}(G, M)$.
		\State For each $\job_j \in J'$: assign $\job_j$ to the machine where $p_{i, j}$ is the smallest. 
		\State Reconstruct the schedule from the assignments.
	\end{algorithmic}
	\caption{$2$-approximate algorithm for $R2|G = \bipartite|\cmaxcost$.}
	\label{alg:2apxUnrelated}
	\end{algorithm}

	\newcounter{countertheoremnunrel}
	\setcounter{countertheoremnunrel}{\value{theorem}}
	\begin{theorem}
		\cref{alg:2apxUnrelated} is $2$-approximate $\Osymbol(n)$-time algorithm for $R2|G = \bipartite|\cmaxcost$.
		\label{theorem:twoapproxunrel}
	\end{theorem}
	For the proof of the theorem we refer the reader to the appendix.		
	
	Now we are able to use the FPTAS for $Rm||\cmaxcost$ presented in \cite{JansesnPImprovedApproximation2001} to obtain an FPTAS $R2|G = \bipartite|\cmaxcost$.
	\begin{algorithm}[H]
		\begin{algorithmic}[1]
			\Require $G = \bipartite = G_1 \cup \ldots \cup G_c$, $M = \{M_1, M_2\}$	
			\Ensure A $(1 + \epsilon)$-approximate algorithm for $R2|G = \bipartite|\cmaxcost$.
			\State Let $T$ be $\cmaxcost$ of the schedule constructed by \cref{alg:2apxUnrelated}.
			\State $(J' , P', P'') := \cref{alg:unrelatedreduced}(G, M)$.
			\State Construct $\job_{n+k+1}$ and $\job_{n+k+2}$ as follows:
			\State \indent Let $p_{1, n+k+1} := \sum_{k = 1}^cP'_k; \text{and } p_{2, n+k+1} := 2T$.
			\State \indent Let $p_{2, n+k+2} := \sum_{k = 1}^cP''_k; \text{and } p_{1, n+k+2} := 2T$.
			\State Apply the FPTAS by Jansen and Porkolab \cite{JansesnPImprovedApproximation2001} to $J' \cup \{\job_{n+k+1}, \job_{n+k+2}\}$.
			\State From the constructed schedule reconstruct the schedule for original jobs.
		\end{algorithmic}
		\caption{$(1+\epsilon)$-approximate algorithm for $R2|G = \bipartite|\cmaxcost$.}
		\label{alg:FPTASUnrelated}
	\end{algorithm}

	\begin{theorem}
		\cref{alg:FPTASUnrelated} is an $\Osymbol(n\frac{1}{\epsilon})$-time FPTAS for $R2|G = \bipartite|\cmaxcost$.
		\label{theorem:FPTAS}
	\end{theorem}
	\begin{proof}
	Notice that in any reasonable schedule the minimum processing times on $M_1$ and minimum processing times on $M_2$ are respected by the assignment of the jobs $\job_{n+k+1}, \job_{n+k+2}$.
	The assignment is enforced by the processing times.
	Hence, the schedule produced can be reinterpreted for the original problem.
	Notice that, for any schedule corresponding to the original instance there is a schedule for the modified instance with exactly the same $\cmaxcost$.
	On the other hand any schedule of the prepared jobs can be interpreted as a schedule of the same $\cmaxcost$ of the original jobs.
	\end{proof}
	
	We have the following inapproximability results.
	\begin{theorem}[\cite{Pikies2020scheduling}]
		There is no $c$-approximate algorithm for $R|G = \completebipartite|\cmaxcost$. 
	\end{theorem}
	However, the number of the machines constructed in the reduction is not bounded.
	Hence, for any fixed $m \ge 3$ we prove the following, by an observation that we could easily solve $\OnePreExt$ using any reasonable algorithm for $Rm|G = \bipartite|\cmaxcost$.

	\begin{theorem}
		For $m \ge 3$ there is no $\Osymbol(n^{b}p_{\max}^{1-\epsilon})$-approximate algorithm for $Rm|G = \bipartite|\cmaxcost$, for any $b > 0$, $\epsilon > 0$, unless $\PClass = \NPClass$.
	\end{theorem}
	\begin{proof}
		Assume that there exists $(c \cdot n^{b}p_{\max}^{1-\epsilon})$-approximate algorithm.
		Again consider \OnePreExt for $3$ colors and bipartite graphs.
		Let $(G, (v_1, v_2, v_3))$ be an instance of this problem.
		Let us construct and instance of $Rm|G = \bipartite|\cmaxcost$:
		Let $d \ge 1$ be an integer value, a function of $n, b, \epsilon$ that we specify later. 
		Let $p_{i, j} = d$, for $j \in [3]$ and $i \in [3] \setminus \{j\}$.
		Let $p_{j, j} = 1$, for $j \in [3]$.
		Let $p_{i, j} = 1$, for $i \in [3]$ and $j \in [n] \setminus [3]$.
		Let $p_{i, j} = d$, for $i \in [m] \setminus [3]$ and $j \in [n]$.

		Assume that the instance of \OnePreExt is 'YES'-instance.
		Then there exists schedule of $\cmaxcost \le n$, by extension of the coloring.
		Hence the algorithm applied to 'YES'-instance has to return a schedule $S$ with $\cmaxcost(S) \le c \cdot n^{b+1}d^{1-\epsilon}$.
		Assume that the instance of \OnePreExt is a 'NO'-instance.
		Then in any schedule $\cmaxcost \ge d$.
		
		Hence if $d > (cn^{b+1})^{\frac{1}{\epsilon}}$, for example $d = \ceil{(cn^{b+1})^{\frac{1}{\epsilon}}}+1$, which can be encoded on a polynomial number of bits, then we could distinguish between the instances.
	\end{proof} 
\section{Open Problems}
	\label{section:open}
	Albeit we explored the problem a bit, there are still a few outstanding problems.
	We would like to point out that the research on uniform machines shall be continued.
	By the simple observations we were able to prove that there is no good method in general case.
	However, in a real world we would like to know how to schedule the jobs on a given \emph{fixed} set of machines.
	Hence it is interesting what is the best possible approximation ratio for a given sequence, finite or not, of machine speeds.
	The author of \cite{bodlaender1994scheduling} proved that for $m$ machines of equal speeds the best possible, and easily achievable, approximation ratio is $2$.
	Hence, we think that constructing a method calculating the best possible approximation ratio, under $\PClass \neq \NPClass$ for a given sequence of speeds is a very interesting open problem.
	Moreover, we think that constructing a kind of scheme, i.e. family of algorithm achieving this approximation ratio will be also very interesting.
	The second observation is that our algorithm for random bipartite graphs could be substantially enhanced.
	For example for $p(n) = \osymbol(\frac{1}{n})$ it could be improved, by better assigning the isolated jobs and using them to "balance" the schedule.
	However, the key to enhance the proposed algorithms seems to be an understanding when there are substantial classes spanning two parts of graphs.
	Considering unrelated machines the question is if there exists any reasonable and interesting subcase of these machines that can be approximated well.

\bibliography{bibliography_isaac}
\newpage
	\section*{Appendix}

	\subsection*{Proof of \cref{theorem:Q2CmaxOpt}}
	\setcounter{theorem}{\thecountertheoremuniform}
	\begin{theorem}
		There exists $\Osymbol(n^3)$ time algorithm for $Q2|G = \bipartite, p_j=1|\cmaxcost$.
	\end{theorem}
	\begin{proof}
		We assume that the input for an FPTAS for $R2|G = \bipartite|\cmaxcost$ is given by a matrix of size $2 \times n$ of processing times.
		The algorithms proceeds by finding an optimal distribution of jobs between the machines.
		Notice that there are $\Osymbol(n)$ pairs $(n_1, n_2)$ such that $n_1 \in \N, n_2 \in \N$ and $n_1 + n_2 = n$.
		The case $n_1 = n, n_2 = 0$ is trivial to check; hence assume that $n_1 \neq 0$ and $n_2 \neq 0$.
		Now, for each pair we prepare the following instance, to check if such a distribution of jobs is feasible.
		Prepare $2$ machines $M_1, M_2$.
		Construct exactly the same $G = (J, E)$ as provided for the instance of $Q2|G = \bipartite, p_j=1|\cmaxcost$.
		However, let the processing time be $p_{i, j} = \frac{n_1 n_2}{n_i}$.
		Now, assume that the jobs can be scheduled in such a way that $M_1$ and $M_2$ receive exactly $n_1, n_2$ jobs, respectively. 
		Otherwise, $\optcmaxcost = \frac{n_1 \cdot n_1n_2}{n_1} = \frac{n_2 \cdot n_1n_2 }{n_2} = n_1 n_2$.
		Now, observe that if an FPTAS for $R2|G = \bipartite|\cmaxcost$ is applied on the prepared instance with $\epsilon = \frac{1}{n}$, then it has to return a solution corresponding to exactly the desired schedule.
		I.e. machine $M_1$ has to do $n_1$ jobs and machine $M_2$ $n_2$ jobs.
		To see this, for a contradiction assume that in the schedule constructed by the FPTAS there are more jobs assigned to machine $M_1$.
		Then we have for such a schedule $S$ $\cmaxcost(S) \ge \frac{(n_1 + 1)\cdot n_1n_2}{n_1} = n_2(n_1 + 1) = n_2n_1(1 + \frac{1}{n_1})$.
		Hence $\frac{\cmaxcost(S)}{\optcmaxcost} \ge 1 + \frac{1}{n_1} > 1 + \frac{1}{n}$, due to the fact that $n_1 \in \{1, \ldots, n-1\}$.
		Similar observations hold for $M_2$.
	\end{proof}
	
	\subsection*{Proof of \cref{lemma:qapxcomplexity}}
	\setcounter{lemma}{\thecounterqapxcomplex}
	\begin{lemma}
		\cref{algorithm:SqrtPtotalForQBipartiteCmax} has time complexity $\Osymbol(|J|^2 + |J||E| + |M|\log |M|)$.
	\end{lemma}
	\begin{proof}
	\begin{itemize}
		\item First we have to check if the vertices of size at least $\sqrt{\psum}$ form independent set, in time $\Osymbol(|J|^2)$.
		\item If yes, then we have to find an independent set of highest weight in the graph formed by removal of the closed neighborhood of the independent set.
		This can be done by finding a minimum $S-T$ cut with a flow network corresponding to the bipartite graph.
		It can be done in $\Osymbol(|J||E|)$-time \cite{OrlinMaxFlowsInMN13} by first finding max flow, then transforming it to a minimum cut in $\Osymbol(|J| + |E|)$ time.
		\item Generalized to weighted graphs inequitable coloring can be calculated in $\Osymbol(|J| + |E|)$ time, by aggregating the colorings of the components appropriately.
		\item Calculating $\optcmaxcostlb$ can be done in $\Osymbol(|M|\log |M|)$.
		The time for the first condition is estimated as follows.
		First calculate "relaxed" time (assume that the jobs can be split arbitrarily).
		Round down the capacities (decreasing the capacities by at most $|M|$).
		Construct a heap composed of $|M|$ entries, composed of times when the capacity of the machine increases to the next integer number.
		Take the smallest time, replace the entry in heap by the next time when the capacity of machine corresponding to entry increases to the next integer.
		Repeat as necessary.
		It can be done in $\Osymbol(|M|\log |M|)$ time, due to the heap reorganization.
		The time for the second condition can be calculated similarly.
		The third condition is trivial.
		\item Calculating $k$ and $k'$ is trivial and the the scheduling of jobs on the respective sets of machine can be done in $\Osymbol(|J|+|M|)$ time.
	\end{itemize}
	Together this gives the desired complexity.
	\end{proof}

	\subsection*{Proof of \cref{lemma:vtwotonminusalpha}}
	\setcounter{lemma}{\thecountervtwotonminusalpha}
	\begin{lemma}
		Let $p(n) = \frac{a}{n}$, for some $a > 0$.
		Then for an inequitable coloring $(V_1', V_2')$ of $\gnnp$, $\frac{|V_2'|}{n - \alpha(\gnnp)} \le 1.6$  a.a.s.
	\end{lemma}
	\begin{proof}
		Using previous observations we have that $|V_2'| \le n(1 - (1 + \frac{a}{n})^n) + \osymbol(n)$.
		Similarly we have $n - \alpha(\gnnp) \ge n(1-e^{(e^{-a}-1)})$.
		Hence we have
		\[
		\frac{|V_2|}{n - \alpha(\gnnp)} \le \frac{1 - (1 - \frac{a}{n})^n}{1 - e^{e^{-a}-1}} + \osymbol(1).
		\]
		Therefore, for big $n$ we have
		\[
		\lim\limits_{n \rightarrow \infty} \frac{1 - (1 - \frac{a}{n})^n}{1 - e^{e^{-a}-1}} = \frac{1 - \frac{1}{e^a}}{1 - e^{e^{-a}-1}}.
		\]
		It means that the ratio depends only on $a$ in a way not depending on $n$.
		We would like to find $a$ such that this value is highest possible.
		We do this by a high-school mathematical analysis.
		First we check the derivative,
		\[
		\frac{d}{da} \frac{1 - \frac{1}{e^a}}{1 - e^{e^{-a}-1}} = \frac{e^{-a}}{1 - e^{e^{-a}-1}} - \frac{e^{-a+e^{-a}-1}(1-e^{-a})}{(1-e^{e^{-a}-1})^2}.
		\]
		It is unclear if it is nonnegative for $a \in (0, \infty)$, let us observe that it is nonnegative when 
		\[
		1 -2e^{e^{-a}-1} + e^{e^{-a}-1-a} \ge 0.
		\]
		Hence let us calculate the derivative of left hand side expression 
		\[
		\frac{d}{da}(1 -2e^{e^{-a}-1} + e^{e^{-a}-1-a})  = e^{-2a + e^{-a}-1}(e^{a}-1),
		\]
		which is nonnegative.
		Hence let us calculate the limit of
		\[
		\lim\limits_{a \rightarrow 0} (1 -2e^{e^{-a}-1} + e^{e^{-a}-1-a} )= 0.
		\]
		It means that $1 -2e^{e^{-a}-1} + e^{e^{-a}-1-a}$ is nonnegative for $a \in (0, \infty)$, and it means that derivative of $\frac{1 - \frac{1}{e^a}}{1 - e^{e^{-a}-1}}$ is nonnegative.
		Hence let us check the limit at infinity
		\[
		\lim\limits_{a \rightarrow \infty} \frac{1 - \frac{1}{e^a}}{1 - e^{e^{-a}-1}} = \frac{e}{e-1} < 1.6.
		\]
	\end{proof}

	\subsection*{Proof of \cref{theorem:twoapproxunrel}}
\setcounter{theorem}{\thecountertheoremnunrel}
\begin{theorem}
	\cref{alg:2apxUnrelated} is $2$-approximate $\Osymbol(n)$-time algorithm for $R2|G = \bipartite|\cmaxcost$.
\end{theorem}
\begin{proof}
	First consider \cref{alg:unrelatedreduced}.
	The algorithm consists of an observation that for any connected component $G_k$ the jobs in a partition of the component can be merged into a single one. 
	Then parts $V_1^k, V_2^k$ of component $G_k$ can be schedule in two ways.
	Either incurring processing times $p_{1,1}^*$ on $M_1$ and $p_{2,2}^*$ on $M_2$, or $p_{1,2}^*$ on $M_1$ and $p_{2,1}^*$ on $M_2$.
	If $p_{1,1}^* \le p_{2,1}^*$ and $p_{2,2}^* \le p_{1,2}^*$ ($p_{2,1}^* \le p_{1,1}^*$ and $p_{1,2}^* \le p_{2,2}^*$), then scheduling $V_1^k$ on $M_1$ and $V_2^k$ on $M_2$ ($V_1^k$ on $M_2$ and $V_2^k$ on $M_1$) is clearly superior to the other assignment.
	Otherwise, in any case on the machines processing times $\min\{p_{1,1}^*, p_{1,2}^*\}$ and $\min\{p_{2,1}^*, p_{2,2}^*\}$ are to be accrued on $M_1$ and $M_2$, respectively.
	It means that scheduling of such component can be reduced to a decision of incurring $\max\{p_{1,1}^*, p_{1,2}^*\} - \min\{p_{1,1}^*, p_{1,2}^*\}$ processing time on $M_1$ or $\max\{p_{2,1}^*, p_{2,2}^*\} - \min\{p_{2,1}^*, p_{2,2}^*\}$ processing time on $M_2$ in addition to these minimums.
	To see this:
	\begin{itemize}
		\item Assume that $\max\{p_{1,1}^*, p_{1,2}^*\} = p_{1,1}^*$, then $\max\{p_{2,1}^*, p_{2,2}^*\} = p_{2,1}^*$.
		In this case the first decision is equivalent to scheduling $V_1^k$ on $M_1$ and $V_2^k$ on $M_2$.
		This incurs $\max\{p_{1,1}^*, p_{1,2}^*\} - \min\{p_{1,1}^*, p_{1,2}^*\} + \min\{p_{1,1}^*, p_{1,2}^*\} = \max\{p_{1,1}^*, p_{1,2}^*\}$ on $M_1$ and $p_{2,2}^* = \min\{p_{2,1}^*, p_{2,2}^*\}$ on $M_2$.
		The second decision is equivalent to scheduling $V_1^k$ on $M_2$ and $V_2^k$ on $M_1$.
		\item Assume that $\max\{p_{1,1}^*, p_{1,2}^*\} = p_{1,2}^*$, then $\max\{p_{2,1}^*, p_{2,2}^*\} = p_{2,2}^*$.
		In this case the first decision is equivalent to scheduling $V_1^k$ on $M_2$ and $V_2^k$ on $M_1$.
		The second decision is equivalent to scheduling $V_1^k$ on $M_1$ and $V_2^k$ on $M_2$.
	\end{itemize} 
	
	To see that the schedule constructed above is $2$-approximate notice that for each component the minimum additional processing time is chosen.
	Hence the sum of additional processing times assigned $T_{extra}$ is minimum.
	For each of the machines also the minimum processing times is incurred $T_{1}, T_{2}$, respectively, as in any schedule.
	Hence, the length of the schedule is at most $\max\{T_1, T_2\} + T_{extra}$.
	On the other hand in any schedule its length is at least $(T_1 + T_2 + T_{extra})/2$.
	Clearly, the $\cmaxcost$ of the constructed schedule is at most $2$ times the optimum.
	
	About the procedure to reconstruct the original jobs from the constructed jobs:
	Consider the values $P'$ and $P''$ in pairs, where $k$-th pair is $(P'[k], P''[k])$.
	If the pair is associated with dummy job it means that the processing times of the jobs in pairs are directly corresponding to parts of some component.
	If the $k$-th is associated with non-dummy job, consider where this $k$-th constructed job was assigned.
	If it was assigned to $M_1$, then it means that the processing time on $M_1$ can be divided into $\max\{p_{1,1}^*, p_{1,2}^*\}$ of component $G^k$ and processing time of other parts.
	Similarly, the processing time on $M_2$ is $\min\{p_{2,1}^*, p_{2,2}^*\}$ and processing time of other parts.
	Notice that this corresponds to assignment of component's parts to the machines.
	Similarly, if it was assigned to $M_2$.
\end{proof}	

\end{document}